\documentclass[11pt]{article}

\usepackage{cite}
\usepackage{amssymb,amsfonts,amsthm,amsmath}
\usepackage{graphicx}
\usepackage{xspace}
\usepackage{float}
\usepackage{url}

\def\doctype{2}

\def\tsubmission{2}
\ifnum\doctype=\tsubmission
	\newcommand{\full}[1]{}
	\newcommand{\submit}[1]{#1}
\else
	\newcommand{\full}[1]{#1}
	\newcommand{\submit}[1]{}
\fi

\newtheorem{theorem}{Theorem}
\newtheorem{lemma}[theorem]{Lemma}

\theoremstyle{definition}
\newtheorem{definition}{Definition}
\newtheorem{corollary}[theorem]{Corollary}

\newcommand{\etal}{et al.\ }
\newcommand{\eps}{\epsilon}

\newcommand{\cM}{\mathcal{M}}

\newcommand{\initOneLiners}{%
    \setlength{\itemsep}{0pt}
    \setlength{\parsep }{0pt}
    \setlength{\topsep }{0pt}
}

        {\hspace*{\fill}$\Box$\par}

\usepackage[linesnumbered, ruled]{algorithm2e}
\SetKwRepeat{Do}{do}{while}%

\begin{document}
%
% paper title
% can use linebreaks \\ within to get better formatting as desired
\title{An $O(\log \log m)$-competitive Algorithm for Online Machine Minimization}

% author names and affiliations
% use a multiple column layout for up to three different
% affiliations
%\author{\IEEEauthorblockN{Sungjin Im, Bejamin Moseley, Kirk Pruhs and Clifford Stein}
%\IEEEauthorblockA{Electrical Engineering and Computer Engineering\\
%University of California \\
%Merced, CA 95343\\
%Email:  \{mnaghshnejad, sim3, msinghal\}@ucmerced.edu
%}
%}

% conference papers do not typically use \thanks and this command
% is locked out in conference mode. If really needed, such as for
% the acknowledgment of grants, issue a \IEEEoverridecommandlockouts
% after \documentclass

% for over three affiliations, or if they all won't fit within the width
% of the page, use this alternative format:
% 

\author{
Sungjin Im\thanks{Electrical Engineering and Computer Science, University of California, Merced, CA 95344, USA. Sungjin Im was supported in part by NSF grants CCF-1617653 and CCF-1409130. } \and
Benjamin Moseley\thanks{Washington University in St. Louis, 
	St. Louis, MO 63130, USA, Benjamin Moseley was supported in part by a Google Research Award, a Yahoo Research Award and NSF Grant CCF-1617724.
} \and
Kirk Pruhs\thanks{Department of Computer Science, 
University of Pittsburgh, Pittsburgh, PA 15260, USA. Kirk Pruhs was supported in part by NSF grants CCF-1421508 and CCF-1535755, and an IBM Faculty Award.} \and
Clifford Stein\thanks{Department of Industrial Engineering and Operations Research, Columbia University, New York, NY 10027, USA. Clifford Stein was supported in part by NSF grant CCF-1421161.}
}

% use for special paper notices
%\IEEEspecialpapernotice{(Invited Paper)}

% make the title area
\maketitle

\begin{abstract}
%\boldmath

This paper considers the online machine minimization problem, a basic real time scheduling problem.  The setting for this problem consists of $n$ jobs that arrive over time, where each job has a deadline by which it must be completed. 
%All jobs must be completed by their deadline by being processed preemptively on a set of identical machines. 
The goal is to design an online scheduler that feasibly schedules the jobs on a nearly minimal number of machines.
An algorithm is $c$-machine optimal if the algorithm will feasibly schedule a collection of jobs on $c\cdot m$ machines if there exists a feasible schedule on $m$ machines. 
For over two decades the best known result was a $O(\log P)$-machine optimal algorithm,
where $P$ is the ratio of the maximum to minimum job size.  In a recent breakthrough,  a $O(\log m)$-machine optimal algorithm was given.  
In this paper, we exponentially improve on this recent result by giving a $O(\log \log m)$-machine optimal algorithm.

%We give an online algorithm that feasibly schedules all jobs by their deadlines on $O(m \log \log m)$ machines if those jobs are feasibly schedulable on $m$ machines.
\end{abstract}
% IEEEtran.cls defaults to using nonbold math in the Abstract.
% This preserves the distinction between vectors and scalars. However,
% if the conference you are submitting to favors bold math in the abstract,
% then you can use LaTeX's standard command \boldmath at the very start
% of the abstract to achieve this. Many IEEE journals/conferences frown on
% math in the abstract anyway.

% no keywords

% For peer review papers, you can put extra information on the cover
% page as needed:
% \ifCLASSOPTIONpeerreview
% \begin{center} \bfseries EDICS Category: 3-BBND \end{center}
% \fi
%
% For peerreview papers, this IEEEtran command inserts a page break and

\section{Introduction}

%\subsection{Our Results}

In a typical real time scheduling environment, there is a collection of jobs that arrive over time. This collection of jobs could be generated by a task system.   Each job has a processing time and a deadline, and must be processed by its deadline.  In such a setting, there are typically two types of results in the literature.  One type of result
is the design of a scheduler, and an analysis  that shows that this scheduler can complete all jobs by their deadlines if the job instance satisfies certain conditions.  The other type of result is the design of a feasibility test, and an analysis  that shows that this test will either determine whether a particular scheduler will feasibly schedule any job instance that might arise from  a particular task system on some collection of machines, or determine that there is some job instance that the scheduler will not feasibly schedule on some collection of machines.

One classic scheduling result, of the first type, is that the Earliest Deadline First (EDF) scheduling algorithm is optimal for deadline scheduling
on one machine. That is, given a collection of jobs for which there exists a feasible schedule on one machine, EDF will feasibly schedule that collection of jobs on one machine.  This result is one of the main reasons EDF is widely used in the real time scheduling literature~\cite{LiuL73,BaruahF06,BaruahF07,ChenC11,AhujaLM16,AnderssonT07}.

Meeting all deadlines becomes more challenging when the jobs can be scheduled on a set of $m$ identical machines.  It is known that no optimal online algorithm exists for more than one machine~\cite{DertouzousMok}. That is, for every online scheduler, and for every $m > 1$, there is a collection of jobs that is feasibly schedulable on $m$ machines, but that this scheduler will not feasibly schedule on $m$ machines. This impossibility result has naturally led to a line of research that involves seeking online algorithms that are near-optimal.

The type of near-optimal algorithm that this paper is concerned with is a {\em $c$-machine optimal algorithm}.  We say that an algorithm is $c$-machine optimal if the algorithm will feasibly schedule a collection of jobs on $c\cdot m$ machines if there exists a feasible schedule on $m$ machines.   
The goal of this line of research is to determine how small a machine augmentation parameter $c$ is attainable by an online algorithm.
Determining whether there exists an $O(1)$-machine optimal\footnote{ A function $g(x)$ is $O(f(x))$ if there exists a constant $c$ and any value $x_0$ such that $g(x) \leq cf(x)$ for all $x \geq x_0$.  In particular, an $O(1)$-machine optimal algorithm uses at most $cm$ machines for some constant $c$.} algorithm is considered to be a big open problem in this line of research~\cite{PSTW,PruhsST04}.  

The concept of a $c$-machine optimal algorithm can be related to scheduling task systems in real-time scheduling as follows. If there exists an algorithm that can feasibly schedule jobs generated by a task system on $m$ machines then a $c$-machine algorithm will feasibly schedule the jobs from the task system on $cm$ machines.  In particular,  a $c$-machine algorithm can be used to schedule an infinite set of jobs generated by a task system feasibly on $cm$ machines so long as some algorithm  can feasibly schedule the jobs on $m$ machines.

An important parameter of a job is its relative laxity, which is the job's laxity divided by the length of its lifespan. (The length of its lifespan is its deadline minus its release time,  and
its laxity is the length of its lifespan minus its size.) It is relatively straightforward to observe that if all jobs have relative laxity $\Omega(1)$\footnote{ A function $g(x)$ is $\Omega(f(x))$ if there exists a constant $c$ and any value $x_0$ such that $g(x) \geq cf(x)$ for all $x \geq x_0$.},  then EDF is $O(1)$-machine optimal. See \cite{ChenMS16} for details. Unfortunately, the problem is much more challenging when jobs have smaller relative laxity.  For over two decades the best known result, when there is no restriction on the laxity, was a $O(\log P)$-machine optimal algorithm, where $P$ is the ratio of the maximum to minimum job size \cite{PSTW}.    Essentially, Phillips \etal \cite{PSTW} observed that $O(\log P)$-machine augmentation trivially reduces the general problem to the easy special case that all jobs have almost the same size.   The bound of $O(\log P)$ also has the disadvantage of being dependent on the input data; bounds that are independent of the input data are much stronger.

In a recent major advance, Chen \etal \cite{ChenMS16} gave a novel online algorithm, which we call CMS after the authors' initials.  Their analysis showed that their algorithm is
$O(\log m)$-machine optimal for jobs with relative laxity less than  $1/2$.  The algorithm and analysis are somewhat complex, but the  underlying intuition of the CMS algorithm design is to prioritize jobs that have used the largest fraction of their original laxity. Thus, by running EDF on jobs with relative laxity more than $1/2$ on half the machines, and by running the CMS algorithm on jobs with relative laxity at most $1/2$ on half of the machines, the work \cite{ChenMS16}
obtains an $O(\log m)$-machine optimal algorithm for arbitrary instances. The work of \cite{AzarCohen2017} improved on this slightly by observing that one can combine EDF and the CMS algorithm  somewhat more cleverly to obtain a $O(\frac{\log m}{\log \log m})$-machine optimal algorithm. 

%Kirk writes: I thought this was redundant with previous and future stuff. 
%The question remains, is there an $O(1)$-machine optimal algorithm?  Resolving this question is both of practical and theoretical interest.   It has been a major target in the scheduling algorithm design community.  In the applied real time community, such an algorithm is necessary to design tighter algorithms for multiple machine feasibility tests.
%\textbf{SJ: Do we mean online feasiblity test? In other words, should the algorithm be online? If we're talking about offline feasiblity test, there is a simple offline algorithm that tests a given of jobs are scheduleable on m machines.}

\subsection{Our Results}

Our main result is an exponential improvement on the machine augmentation parameter achieved in  \cite{ChenMS16,AzarCohen2017}. We give a new algorithm (called Algorithm A in this paper) and analysis showing that Algorithm A is $O(\log \log m)$-machine optimal. 

Our algorithm is constructed from several  building blocks. The initial insight that led to our main result, and the first building block, is the observation that the algorithm Shortest Job First (SJF) is $O(1)$-machine optimal if all jobs have approximately the same relative laxity. More precisely, we show the following.  This proof of the following lemma is given later in the paper.

\begin{lemma}
\label{lem:initial}
For a collection of jobs with relative laxities in the range 
$[\lambda_1, \lambda_2] \subseteq (0, 1/2]$, Shortest Job First is $O( \log_{1 / \lambda_2} \lambda_2/\lambda_1 )$-machine optimal for any $\lambda_1, \lambda_2$. 
\end{lemma}

The second building block is that an implication of Lemma \ref{lem:initial} is that 
SJF is $O(1)$-machine optimal if the relative laxities
of the jobs lie in the range $[1/2^{2^{i+1}}, 1/2^{2^i}]$, for some  $i \ge 1$. 
This leads to an algorithm with machine augmentation doubly logarithmic in the inverse of the minimum relative laxity of any job. More precisely, we show that:

\begin{lemma}
\label{lem:second}
Consider some instance where the relative laxities lie in the range $[1/R, 1/2]$.
There is a $O(\log \log R)$-machine
optimal algorithm.
\end{lemma}

\begin{proof}
Partition the jobs into $\lg \lg R$ different groups, 
where the jobs in group $i$ have relative laxities in the range $[1/2^{2^{i+1}}, 1/2^{2^i}]$. Use SJF to run the jobs in each group on $O(m)$ machines dedicated to that group.   By Lemma~\ref{lem:initial} the total number of machines per group required is $O( m \log_{ 2^{2^i}} (2^{2^{i+1}}/2^{2^{i}}) ) = O( m \log_{ 2^{2^i}} 2^{2^{i}} ) = O(m)$.  Thus, at most $O(m\lg \lg R)$ machines are needed to ensure all job are completed by their deadlines.
\end{proof}

The final building block is the observation that by tweaking the analysis in \cite{ChenMS16}, one can show that the CMS algorithm is $O(1)$-machine optimal if the relative laxities of the jobs are all at most $1/m$. The jobs with relative laxities at least 1/2 are scheduled on separate $O(m)$ machines using EDF as in the work by Chen \etal \cite{ChenMS16}. Putting all these building blocks together, we obtain our main result, stated in Theorem \ref{thm:main}, that the following algorithm $A$ is $O(\log \log m)$-machine optimal.  See Section~\ref{sec:alg} for full the description of the algorithm $A$.

%The main result of this paper is the following.
\begin{theorem}
	\label{thm:main}
	There  is a $O(\log  \log m)$-machine optimal algorithm. 
\end{theorem}

%Our results have further implications regarding another type of resource augmentation.  So far, we have assumed one type of resource augmentation.  That is, it has been assumed that the algorithm is given extra machines over the optimal algorithm and this is known as machine augmentation.  

The results in this paper  have further implications regarding another type of resource augmentation, namely
speed augmentation, which is commonly
 used, either instead of, or in conjunction with, machine augmentation. 
In our context, an $s$-speed $c$-machine optimal algorithm would feasibly 
schedule a job instance on $cm$ machines of speed $s$ if this job instance is feasibly schedulable on $m$ machines of speed 1. Speed augmentation is widely used for designing near optimal algorithms, and in  corresponding feasibility tests \cite{LiuL73,BaruahF06,BaruahF07,ChenC11,AhujaLM16,AnderssonT07}.  
The best combined speed and machine augmentation result comes from the paper \cite{LamTo1999}. This paper  showed the existence of a $(1+\epsilon)$-speed $O(\frac{1}{\epsilon})$-machine optimal algorithm. As a corollary of our main result, we can give a doubly-exponential  improvement on the trade-off of speed and machine augmentation, stated in Corollary \ref{cor:machinespeed}:

\begin{corollary}
\label{cor:machinespeed}
There is a $(1+\epsilon)$-speed $O(\log \log \frac{1}{\eps})$-machine optimal algorithm for any $\eps >0$.
\end{corollary}

The proof of Corollary \ref{cor:machinespeed} follows from the observation that
on a $1+\epsilon$ speed machine, every job (that is feasibly schedulable on a speed 1 machine) has relative laxity at least $\frac{\epsilon}{1+\epsilon}$. 

\noindent \textbf{Application to Task Systems:} We now comment briefly on the application of these results to periodic/real-time scheduling. On the positive side, as these results apply to all jobs instances, they apply to job instances that arise from periodic task systems. A corollary to our results is that if every job instance arising from a particular  task system can be scheduled on $m$ machines, then algorithm $A$ will schedule every job instance arising from this task system on $O(m \log \log m)$ machines. On the negative side, in the context of real-time scheduling, one generally also wants a corresponding efficient feasibility test that matches the optimality result. In our context, such a test would take as input  a task system and a number of machines $m$, and would determine whether the algorithm $A$
will feasibly schedule any job instance that might arise from this task system on 
$O(m \log \log m)$ machines, or determine that there is some job instance that might arise from this task system that is not feasibly schedulable on $m$ machines. As an example using speed augmentation, the paper \cite{PSTW} showed that EDF is 2-speed optimal, and the paper
\cite{BonifaciMS12} extended this speed-augmentation optimality result to  an efficient feasibility test for EDF for speed 2 processors. That is, this test determines   whether every job instance arising
from a task system will be feasibly scheduled by EDF on $m$ processors of speed 2, or determines that some job instance arising from the tasks system is not feasibly schedulable on $m$ processors of speed 1. Unfortunately, machine augmentation is more combinatorially complicated than speed augmentation, and we do not yet know how to extend our machine-augmentation optimality result to an efficient feasibility test. The best we can say is that an optimality result is the first step toward achieving a feasibility test.  %Also on the negative side, the algorithm $A$ was designed to achieve a theoretical result, and is certainly not natural or practical. 

\section{Further Related Work}
	\label{sec:related}

\paragraph{Non-Periodic Scheduling}
In addition to solving the general deadline problem, the paper \cite{ChenMS16} handled the special cases of laminar and agreeable deadlines\footnote{In the laminar case, for any pair of jobs $i$ and $j$, either the two jobs lifespans are disjoint or one job's lifespan fully contains the other job's. In the agreeable deadline case, if job $i$ is released earlier than job $j$, then $i$ has a deadline no later than $j$.}, showing
 that their algorithm is $O(1)$-machine optimal for these job instances. Following up on this  result, Chen \etal \cite{spaaMS16} shows that there is no non-migratory $O(1)$-machine optimal algorithm. A non-migratory algorithm schedules each job on a unique machine. For the definition of laminar and agreeable deadlines, see  \cite{ChenMS16}.

%There are several results that consider both speed and machine augmentation.  Most notably, Lam et. al
The paper \cite{LamTo1999} gave an online algorithm that is 
$(2 - \frac{2 (m-1) + mp}{(m-1)(m+1) + mp})$-speed $(m+p)$-machine optimal and they  also give a slightly weaker trade-off analysis for the EDF algorithm.
The works \cite{PSTW,LamTo1999} gave lower bounds showing that there is no $(1 + o(1))$-machine optimal algorithm.

\paragraph{Real-time/Periodic Scheduling}
For a survey of standard terminology and notable results for real-time scheduling see \cite{DavisB11}. Most of the related results in the real-time literature are about
partitioned scheduling, where all jobs emanating from the same task have to run on the same machine.  The paper \cite{Mok} shows that the problem of deciding whether an implicit deadline task set is feasible on a certain number of machines is NP-hard. The paper \cite{ChenChakraborty} shows that it is NP-hard to differentiate between  implicit deadline task systems that are feasible on 2 machines from those that require 3 machines. 
  The paper \cite{BaruahF06} gives a partition algorithm that guarantees feasibility on speed 3 machines for a constrained deadline task system if there is a partition that is feasible on speed 1 machine.  The paper \cite{ChenChakraborty,Baruah11} provide polynomial-time approximation schemes for some special cases when speed augmentation is adopted. The papers \cite{ChenChakraborty,ChenBansalChakraborty} rule out the existence of asymptotic approximation schemes for certain types of task systems. Finally, the paper \cite{ChenBansalChakraborty} provides polynomial time partitioning algorithms whose approximation ratios are a function of the maximum ratio of the period to the deadline.

\section{Formal Problem Definition and Notations}
	\label{sec:problem}

A (finite) set of $n$ jobs arrive over time to be scheduled on $m$ identical  machines/processors.  These jobs could be  generated by a task system.  Each job $j$ has size $p_j$ and deadline $d_j$. 
 The online scheduler only learns of job $j$ when $j$ arrives at its release time $r_j$. This paper assumes that all job characteristics, $p_j$, $d_j$, and $r_j$, are integers. Each machine can process at most one job at a time, and no job can be processed at the same time on two different machines. The paper considers preemptive migratory scheduling, which means that there is no further restrictions on when and where a job is processed. 
A job $j$ completes if it gets processed for $p_j$ units of time. We say that a job is \textbf{alive} at time $t$ if the job has arrived and hasn't completed at $t$.   Every job $j$ must be processed and completed within their \textbf{lifespan} (processing interval) $I(j) := (r_j, d_j)$. 
 We say that a schedule is feasible if all jobs complete within their lifespans.

A job $j$'s laxity $\ell_j$ is defined as $d_j - r_j - p_j$. In words, job $j$ may spend no more than  $\ell_j$ time steps during its lifespan not being processed in a feasible schedule. Intuitively, one can think of $\ell_j$ as $j$'s budget, and the job has to pay a unit cost out of its budget when it is not processed. 
%A job violates its deadline when its budget goes negative. 
A job $j$'s  \textbf{relative laxity}, denoted as $\rho_j$, is $\ell_j/|I(j)|$, the ratio of the job's laxity to its lifespan length. We say that $j$ \textbf{covers} time $t$ if $t$ is within $j$'s lifespan, i.e. $t \in I(j)$. For a set of jobs $S$,  define $I(S) := \cup_{j \in S} I(j)$. For a finite collection $I$ of disjoint intervals,  let $|I|$ denote the total length of intervals in $I$. 

Let $m^*$ denote the minimum number of machines that admits a feasible offline schedule for a given instance. 
It can be assumeed without loss of generality  that  $m^*$ is known to the algorithm up to a constant factor using a standard doubling trick -- if our scheduling fails due to underestimating $m^*$, we simply double our estimate. See \cite{ChenMS16} for more details.  The algorithm used in this paper will be parameterized by $m^*$.

Let $\alpha$ be a scalar.  We say that job $j$ is \textbf{$\alpha$-loose} if $p_j \leq \alpha | I(j) |$, otherwise it is \textbf{$\alpha$-tight}. We say a job is simply \textbf{loose} if it is $\frac{1}{2}$-loose.
We say that a job is \textbf{very tight} if its relative laxity is at most $1/ m$.
Note that relative laxity and $\alpha$ are defined differently, so a job that is $\alpha$-loose for a large $\alpha$ actually has a small relative laxity.
%Let $R_k$ denote jobs of relative laxities $\in (\frac{1}{2^{2^{k+1}}}, \frac{1}{2^{2^k}}]$.

Whenever a machine becomes available, the algorithm Shortest-Job-First (SJF) chooses to schedule the uncompleted job $j$ with the smallest \emph{original} work, $p_j$. 
Whenever a machine becomes available, the algorithm Earliest-Deadline-First (EDF) chooses to schedule the uncompleted job $j$ with the smallest deadline, $d_j$. 
\section{Algorithm Description}
\label{sec:alg}
\noindent 

\newcommand{\algoA}{\boldsymbol{A}}
\newcommand{\cL}{\mathcal{L}}
\newcommand{\sedf}{\mathsf{edf}}
\newcommand{\ssjf}{\mathsf{sjf}}
\newcommand{\scms}{\mathsf{cms}}

This section describes our $O(\log \log m)$-machine optimal algorithm, which will be denoted as $\algoA$. The algorithm is hybrid and runs several different procedures depending on the relative laxity of jobs.  The algorithm $\algoA$ is parameterized by $m^*$, the minimum number of machines required to feasibly schedule the jobs by any (offline) algorithm. To present our algorithm more transparently, we describe our algorithm assuming that the parameter $m^*$ is known to the algorithm a priori---we will show in Section~\ref{sec:remove} how we can easily remove the assumption by using at most four times more machines. Likewise, we will show in Section~\ref{sec:remove},  $\algoA$ can be implemented without knowing parameters $m_{\sedf}$,  $m_{\ssjf}$ and $m_{\scms}$ that appear in the following; all parameters will be shown to be $O(m^*)$. 

\begin{itemize}
\item Earliest Deadline First (EDF): Jobs with relative laxity at least $1/4$ are scheduled using EDF on $m_{\sedf}$ dedicated machines. 

\item Shortest Job First (SJF): Let $\cL_i$ denote the set of jobs with relative laxity in the range of $(1/2^{2^{i+1}}, 1/2^{2^i}]$ where $i$ is an integer in the range of 
$[1, \lceil \lg \lg m^* \rceil]$; here, $\lg$ has a base of 2. For each $i$, a set of $m_{\ssjf}$ machines, $\cM_i$, are dedicated to processing jobs in $\cL_i$ using SJF.  At any point in time SJF schedules up to $m_{\ssjf}$ jobs with the smallest sizes. 
%Jobs with laxities in the range $[1/2^{2^{i+1}}, 1/2^{2^i}]$ are scheduled using SJF on $O(m)$ dedicated machines. The range of $i$ is $[1, \lg \lg m]$. There are $\lg \lg m$ different copies of SJF running  disjoint groups of jobs on $\lg \lg m$ groups of machines, with $O(m)$ machines per group. 

\item Chen-Megow-Schewoir (CMS): The remaining jobs, which have relative laxity no greater than $1/m^*$, are scheduled using the CMS algorithm \cite{ChenMS16} on $m_{\scms}$ dedicated machines.  The description of the CMS algorithm is given in the next section.
%Jobs with laxities in the rage $[0, 1/m]$ are scheduled using the CMS algorithm \cite{ChenMS16} on $O(m)$ dedicated machines. 
\end{itemize}

Note that EDF, SJF and CMS algorithms are used to process jobs with relative laxities that are high, intermediate, and low, respectively. It is important to note that the three algorithms use disjoint sets of machines. Further, the algorithm separately uses SJF for jobs in each  set $\cL_i$ using a distinct set $\cM_i$ of machines. It is easy to see that $O(m^* \log \log m^*)$ machines are used by our hybrid algorithm $\algoA$ if $m_{\sedf}$,  $m_{\ssjf}$ and $m_{\scms}$ are all $O(m^*)$.

%This completes the algorithm description.  
\subsection{Algorithm Chen-Megow-Schewoir (CMS)}

Since the algorithm CMS is not as well known as EDF or SJF, we give a full description of CMS including its pseudocode. The algorithm CMS takes as input a parameter $m_{\scms}$. The algorithm processes jobs using $m_{\scms}+1$ machines, and either outputs a feasible schedule using $m_{\scms}$ machines or declares failure. The $m_{\scms}+1$ machines are indexed by $1, 2, \cdots, m_{\scms}+1$ in an arbitrary but fixed order. The last machine $m_{\scms} +1$ is forbidden, meaning that the algorithm declares failure if it ever processes a job on the machine. Each job $j$ is initially given a budget equal to its laxity, $\ell_i = d_i  - r_i - p_i$, and its budget is equally distributed to the $m_{\scms}+1$ machines. We emphasize that the budget is never shared between machines.  Let $b_{ji}(t)$ denote $j$'s budget for machine $i$ at time $t$. Note that $b_{ji}(r_j) = \ell_j / (m_{\scms}+1)$ for all $i \in [m_{\scms}+1]$. 

We now describe how the algorithm CMS decides which jobs to schedule and which to delay at $t$ at each fixed time $t$. Consider the incomplete jobs in decreasing order of their arrival times, breaking ties in an arbitrary but fixed order.  When considering job $j$, let $i$ be the least indexed machine a job is not currently being scheduled on at the fixed time $t$. We assign $j$ to machine $i$, which doesn't necessarily mean that $i$ processes job $j$ at the moment. 
 If job $j$ has any budget left for machine $i$, i.e. $b_{ij}(t) > 0$, do \emph{not} process $j$, i.e. delay it, decreasing $b_{ij}(t)$ at a rate of 1 at the instantaneous time $t$. If the budget is empty, i.e., $b_{ij}(t) =0$, schedule job $j$ on machine $i$. After either delaying or processing $j$ on machine $i$, we consider the next incomplete job. As mentioned before, the algorithm CMS declares failure if it ever has the forbidden machine  $m_{\scms}+1 $ process a job. 

The algorithm CMS keeps the same schedule, i.e. schedule exactly the same job on the same machine, until time $t''$ when a new job arrives or a job assigned to machine $i$ completely uses its budget for machine $i$ while getting delayed.  Since the above procedure is invoked only for such events,  it is easy to see that CMS runs in polynomial time. For more details, see the pseudo-codes, Algorithms \ref{cms} and \ref{sub-cms}.

It is worth noting that it could happen that a job uses its budget for machine $i$ before depleting it budgets for lower-indexed machines, $1$, $2$, $\cdots$, $i-1$. Thus, when the algorithm declares failure, that is, processes a certain job $j$ on machine $m_{\scms}+1$ at time $t$, it must be the case that $b_{j, m_{\scms}+1}(t) = 0$, but not necessarily $b_{ji}(t) = 0$ for all $i \in [m_{\scms}]$. 

%job $i$'s $j$th budget has any budget available, do not schedule $i$ and move to job $i+1$.   If the budget is empty, schedule job $i$ on machine $j$ and recurse on job $i+1$. 
%that is the number of machines the algorithm is run on.  This is set to be $O(m)$ for an appropriately chosen constant. Each job $j$ is associated with $m'+1$ budgets of size $\ell_j/ (m'+1)$.  That is, the laxity of a job is partitioned evenly in $m'+1$ groups. 
%The algorithm uses the CMS algorithm of \cite{ChenMS16} as a subroutine. 
%This algorithm is defined as follows. 
%Fix a time $t$. The algorithm decides which jobs to schedule and which to delay at $t$.   Order the machines $m'$ in an arbitrary order $1,2,\ldots, m$.  
% Consider the incomplete jobs in decreasing order of their release dates.  When considering job $i$, let $j$ be the least indexed machine a job is not currently being scheduled on.  If job $i$'s $j$th budget has any budget available, do not schedule $i$ and move to job $i+1$.   If the budget is empty, schedule job $i$ on machine $j$ and recurse on job $i+1$. 

\begin{algorithm}
	 \KwIn{A sequence of jobs arriving online; $m_{\scms}$ machines indexed by $1, 2, \cdots, m_{\scms}$}
	 \KwOut{Either yields a feasible schedule or declares failure; a feasible schedule is always output if $m_{\scms} \geq c_{\scms} \cdot m^*$}
%    \KwInOut{Input}{The input sequence of jobs arriving online}
%  \SetKwInOut{Output}{The input sequence of jobs arriving online}
%  \KwData{The input sequence of jobs arriving online}
%  \KwResult{The input sequence of jobs arriving online}
  $t' = 0$  \footnotesize{// the latest time when the Sub-CMS was called}\normalsize{}\;
  $t = 0$   \footnotesize{ // the current time}\normalsize{}\;
  \While{$A_t := \{j \; | \; r_j \leq t, p_j(t) > 0\} \neq \emptyset$ \newline \mbox{$\quad\quad$  }\footnotesize{//$A_t$: jobs alive at time $t$} \normalsize{}\newline}{
%	\KwDataXX // \footnotesize{$A_t$ is the set of jobs alive at time $t$}\;
    $\psi \leftarrow$ Sub-CMS($A_t$, $m_{\scms}$, $\{b_{ji}(t)\}_{j \in A_t, i \in [m_{\scms}+1]}$, $\{p_j(t)\}_{j \in A_t}$)\;
    \If{$\exists j \in A_t$ such that $\psi(j, t) = m_{\scms} +1$ and $b_{j, \psi(j,t)} = 0$}
    {
    	declare failure and terminate\;
    }
            $\Delta_1 = \min \{ b_{j, \psi(j,t)}(t) \; | \; b_{j, \psi(j,t)} >0, j \in A_t\}$\;
    $\Delta_2 = \min \{ p_j(t) \; | \; b_{j, \psi(j,t)}  = 0, j \in A_t\}$\;
    $\Delta = \min \{ \Delta_1, \Delta_2\}$\;
    $t' = t$\;
    $t = t' + \Delta$\;
     \If{ a new job $j$ arrives before time $t$}   
     {
     	\For{all $i \in  [m_{\scms}+1]$}
	{
		$b_{ji}(r_j) = \ell_j / (m_{\scms} +1)$\;
	}
     	$t = r_j$\;
     }
     \For{all $j \in A_{t'}$ such that $b_{j, \psi(j,t')} >0$}{
     	 $b_{j, \psi(j,t')}(t) = b_{j, \psi(j,t')}(t') - (t - t')$\;
     }
     \For{all $j \in A_{t'}$ such that $b_{j, \psi(j,t')} =0$}{
     	$p_j(t) = p_{j}(t') -  (t - t')$\;
     }
 }
  \caption{Algorithm Chen-Megow-Schewoir (CMS)}
 \label{cms}
\end{algorithm}

\begin{algorithm}
	 \KwIn{$A_t$; machines $1, 2, \cdots, m_{\scms}+1$; $b_{ji}(t)$ for all $j \in A_t$ and $i \in [m_{\scms}+1]$; $p_j(t)$ for all $j \in A_t$.}
	 \KwOut{$\psi(j, t)$ for all all jobs $j \in A_t$.} 
	 Order jobs in $A_t$ in non-increasing order of their arrival times\;
	 $i = 1$\;
	\For{each $j \in A_t$}{
	     $\psi(j, t) = i$\;
	     \If{If $b_{ij} = 0$} {
	     	$i = i+1$\;
	     }
	} 
  \caption{Algorithm Sub-CMS}
   \label{sub-cms}
\end{algorithm}

We now take a close look at CMS taking into account issues arising in its implementation. Algorithm~\ref{cms}, CMS, is described assuming that the first job arrives at time $0$ and no two jobs arrive at the same time. These assumption can be made w.l.o.g. by shifting the time horizon and breaking ties between jobs with the same arrival time in an arbitrary but fixed order. Algorithm~\ref{cms} uses Algorithm~\ref{sub-cms}, Sub-CMS, as a sub-procedure. When CMS calls Sub-CMS, it passes to the sub-procedure the set of alive jobs at the moment, $A_t$, and the number of the given machines, $m_{\scms}$, along with jobs' remaining budgets $\{b_{ji}(t)\}$ and remaining sizes $\{p_j(t)\}$; here $p_j(t)$ denotes $j$'s remaining size at time $t$. Then, Sub-CMS finds an assignment of each job to a machine at time $t$: $\psi(j, t) = i$ implies $j$ is assigned to machine $i$ at time $t$. If $j$ still has some budget left for machine $i$ at the time, i.e., $b_{j, \psi(j,t)} > 0$, then $j$ is delayed (not processed) at the moment, burning its remaining budget for machine $i$; we call such jobs \emph{inactive}. Otherwise, $j$ is processed on machine $i$ and is said to be \emph{active}. We note that the assignment ensures that each machine processes at most one job at the time. The last machine $m_{\scms}$ is `forbidden' in the sense that if the machine processes a job, then CMS declares failure and terminates as shown in Lines 5-7.

Then, CMS computes the next time step when it needs to call Sub-CMS again. There are three types of events that triggers calling Sub-CMS. The first type of  event is when an inactive job $j$ depletes its budget $b_{ji}(t) = 0$ for the machine $b_{ji}(t)$ to which it was assigned at time $t$, which is expected to happen in $b_{ji}(t) = 0$ units of time from $t$ if no other events occur before. By taking the minimum over all inactive jobs, a first-type event occurs in $\Delta_1$ units of time from $t$ if no other types of events occurs before.  The second type of event is when an active job is completed. A second-type even occurs in $\Delta_2$ units of time assuming that no other types of events occur before. Thus, CMS doesn't have to reassign jobs within $\Delta := \min\{\Delta_1, \Delta_2\}$ time steps if no jobs arrive in the future, and thus can increase the current time $t$ to $t + \Delta$, updating $t'$ as well in Lines 11 and 12.

The last type of event is when a new job arrives, which is handled in Lines 13-18. Thus, if CMS receives a new job $j$ to schedule before the next time $t$ it planned to reassign jobs by calling Sub-CMS, then it initialize $j$'s budgets in Line 15 and update $t$ to $r_j$, meaning that it needs to call Sub-CMS right now due to job $j$'s arrival. In Lines 19-21, CMS updates inactive jobs' budgets. Note that $i = \psi(j, t')$ is the machine to which $j$ was assigned and $j$ was delay burning its budget for machine $i$ at a rate of 1 at each time during time interval $[t', t]$. In Lines 22-24, CMS processes all active jobs during the same time interval and update their remaining sizes appropriately. After this update, a new assignment $\psi$ is computed via call to Sub-CMS in Line 4 if there is any alive job.

\subsection{Without Knowing the Minimum Number of Machines Admitting a Feasible Schedule}
	\label{sec:remove}

\newcommand{\shigh}{\mathsf{high}}

We gave a full description of our algorithm assuming that the algorithm knows $m^*$, the minimum number of machines admitting a feasible (offline) schedule, a priori. In this section, we show how we can remove this assumption by a simple trick of doubling the number of machines whenever we realize we underestimated the true value of $m^*$. Specifically, we show the following `conversion' lemma. 

\begin{lemma}
	\label{lem:double}
	Given a $c$-machine-optimal online algorithm that takes $m^*$ as a parameter, we can convert it into one that is $4c$-optimal without using the parameter.
\end{lemma}

To convey the main idea more transparently and illustrate how it is used, we show this theorem for a specific algorithm EDF---however, the proof is completely oblivious to EDF, and therefore, we will have the theorem immediately. Recall that we use EDF to process jobs with relative laxity at least $1/4$, whose set is denoted as $\cL_{high}$. The following theorem shown in  \cite{ChenMS16} (Theorem 2.3),
\begin{theorem}
	If all jobs have relative laxity at least $\rho$, EDF is $1 / \rho^2$-machine optimal. 
\end{theorem}
\noindent
implies that if we run EDF using $c m^*$ machines, where $c = 16$, we can feasibly schedule all jobs in $\cL_{high}$. If we had known $m^*$ from the beginning, we could have dedicated $c m^*$ or more machines from time 0 and we would have successfully scheduled all jobs in $\cL_{high}$. 
%It is known that for any set $J$ of jobs with relative laxity $1/2$ that is schedulable on $m$ machines can be feasibly scheduled on $4m$ machines \cite{ChenMS16}. So, clearly the subset of input jobs with relative laxity $1/2$, which denoted as $J_{\shigh}$ is schedulable on $4m^*$ machines by EDF.  
To avoid using the knowledge of $m^*$, at a high-level, we will partition the time horizon $[0, \infty)$ online into disjoint intervals, $I_1 := [t_0:=0, t_1), I_2 := [t_1, t_2),  \cdots, I_\kappa := [t_{\kappa-1}, t_\kappa = \infty)$. Each interval $I_k$ is associated with $2^{k-1}$ machines that are exclusively dedicated to processing jobs arriving during $I_k$, which we denote as $J(I_k)$. We now describe how we define the times $t_1, t_2, \cdots$ online. Initially, we use only one machine to run EDF. Whenever a new job arrives at time $t$, we simulate EDF's schedule until we complete all jobs that have arrived pretending that no more jobs arrive. If  EDF can feasibly complete all jobs alive at the time using the single machine, then we do nothing. Otherwise, we set $t_1 =t$. As mentioned before, jobs arriving by time $t_1$ are scheduled by the initial single machine. We repeat this recursively: Say the current time $t$ is such that $t \geq t_{k-1}$ but we haven't set $t_k$ yet. When a new job $j$ arrives at time $t$, we simulate EDF's schedule pretending that no more jobs arrive and set $t_k = t$ if it fails to yield a feasible schedule for jobs that have arrived after $t_{k-1}$ using $2^{k-1}$ machines; otherwise we do nothing. 

It is clear that this algorithm always gives a feasible schedule as we use more machines whenever we need more. Therefore, it only remains to show that the number of machines we will have used at the end is not far from $m^*$. We show that $2^{\kappa-1} \leq 2c m^*$, meaning that we use at most $\sum_{k=1}^{\kappa} 2^{k-1} < 2^\kappa \leq 4c m^*$. To see this, for the sake of contradiction suppose $2^{\kappa-1} > 2c m^*$. Thus, we have $2^{\kappa-2} > c m^*$. This means that even if we dedicated more than $c m^*$ machines, we couldn't feasibly schedule all jobs arriving during $I_{\kappa-1}$ and had to use more machines from time $t_{\kappa}$. This is a contradiction to the precondition that all jobs, including those jobs arriving during $I_{\kappa-1}$, are schedulable on $c m^*$ machines. Recalling $c = 16$, we can feasibly schedule all jobs in $\cL_{high}$ using $64 m^*$ machines. As mentioned, this proof is oblivious to the algorithm, hence we have Lemma~\ref{lem:double}. 

Thus, we can use this doubling trick for each run of EDF for jobs in $\cL_{high}$, SJF for each $\cL_i$, $i \in [1, \lceil \lg \lg m^* \rceil]$, and CMS for the other jobs. 
Each run is guaranteed to find a feasible schedule for the jobs it is assigned when using $O(m^*)$ machines. Since we use at most four times more machines for each run by doubling and there are $2 + \lceil \lg \lg m^* \rceil$ runs, we use at most $O( m^* \log \log m^*)$ machines, as desired.

\section{Algorithm  Analysis}
In this section, the theoretical guarantees of algorithm $A$ are shown.  

The main challenge in analyzing the optimality of an online algorithm  is discovering strong lower bounds on $m^*$, the minimum number of machines needed to feasibly schedule a particular job instance. 
%That is, it is essential to discover good lower bounds on $m^*$ so that we may compare the number of machines the algorithm uses to this lower bound to establish the competitiveness of the algorithm.  
In subsection \ref{sec:lb1} and subsection \ref{sec:lb2}, we strengthen two lower bounds found in \cite{ChenMS16}.   

With new lower bounds on $m^*$ in place, Subsection \ref{sec:analysis}
 proves Lemma~\ref{lem:initial}.

 Subsection \ref{sec:theoremadapt} explains why CMS is $O(1)$-machine optimal
 for very tight jobs, and EDF is $O(1)$-machine optimal for loose jobs.

 Our main result, that algorithm $A$ is $O(\log \log m)$-machine optimal, follows by combining these results.

\subsection{First Lower Bound}
	\label{sec:lb1}
	
	This section gives a new lower bound on $m^*$.  To do so, the following important definition originating from \cite{ChenMS16}  is needed. 

\begin{definition}[\cite{ChenMS16}]
	Let $G$ be a set of $\alpha$-tight jobs and let $T$ be a non-empty finite \emph{union} of time intervals. For some $\mu \in \mathbb{N}$ and $\beta \in (0,1)$, a pair $(G, T)$ is called $(\mu, \beta)$-critical if 
\begin{enumerate}
	\item each time $t$ belonging to an interval in $T$ is covered by at least $\mu$ distinct jobs in $G$.  That is, $\mu$ jobs in $G$  include $t$ in their lifespans.
	\item $|T \cap I(j)| \geq \beta \ell_j$ for all $j$ in $G$. 
\end{enumerate}
\end{definition}

Based on this definition, Chen \etal \cite{ChenMS16} gave the following novel lower bound on $m^*$. 

\begin{theorem}[\cite{ChenMS16}]
	\label{thm:kevin-1}
	If there exists a $(\mu, \beta)$-critical pair, then $m^*$ is  $\Omega( \frac{\mu}{\log 1 / \beta})$.
\end{theorem}

% The work \cite{ChenMS16} shows that if $A$ fails to schedule jobs on $m'$ machines, then there exists a $(\mu, 1 / \mu)$-critical pair where $\mu = m'+1$ machines. Hence if  $m' = \Theta(m^* \log m^*)$ machines are used by $A$, the algorithm is guaranteed to feasible schedule  all jobs since otherwise Theorem~\ref{thm:kevin-1} would imply a contradiction. 

	In the following, it is shown that Theorem \ref{thm:kevin-1} can be strengthened.  This was shown independently in \cite{AzarCohen2017}.

\begin{theorem}
	\label{thm:our-1}
	If all jobs are $\alpha$-tight, and there exists a $(\mu, \beta)$-critical pair, then $m^* = \Omega( \frac{\mu}{\log_{1/(1 - \alpha)} 1 / \beta})$.
\end{theorem}

The rest of this section is devoted to proving Theorem \ref{thm:our-1}.  The proof builds on the analysis given in \cite{ChenMS16}.  The proof of Theorem~\ref{thm:kevin-1} in \cite{ChenMS16} establishes the following.

\begin{lemma}[\cite{ChenMS16}]
\label{lem:partiallem}
If there exists a pair $(G,T)$ that is $(\mu, \beta)$-critical then there exists a collection $S_1, S_2, \ldots S_{\lceil 2m^*/\alpha \rceil}$ of pairwise disjoint sets of $\alpha$-tight jobs  where $I(S_1) \subseteq I(S_2) \subseteq  \ldots \subseteq  I(S_{\lceil 2m^*/\alpha \rceil}) $.  Further if $|I(S_{\lceil 2m^*/\alpha \rceil})| \geq \gamma |I(S_1)|$ then  $m^* \geq \Omega(\frac{\mu}{\log_{\gamma} \frac{1}{\beta}})$ for any scalar $\gamma$.
\end{lemma}

After proving this lemma, the proof in \cite{ChenMS16} is completed by showing $\gamma \geq 2$.

%   a collection $S_1, S_2, \ldots S_{\lceil 2m^*/\alpha \rceil}$ of pairwise disjoint sets of $\alpha$-tight jobs such that $I(S_1) \subseteq I(S_2) \subseteq  \ldots \subseteq  I(S_{\lceil 2m^*/\alpha \rceil}) $.  Further it shows that if $|I(S_{\lceil 2m^*/\alpha \rceil})| \geq \gamma |I(S_1)|$ then  $m^* \geq \Omega(\frac{\mu}{\log_{\gamma} \frac{1}{\beta}})$. 
%\begin{lemma}
%\label{lem:kevingeneral}
%If there exists a $(\mu,\beta)$-critical pair, then $m^* \geq \Omega(\frac{\mu}{\log_{\gamma} \frac{1}{\beta}})$.
%\end{lemma}

Given the previous lemma, to prove Theorem \ref{thm:our-1}, it is sufficient to establish a stronger lower bound on $\gamma$, namely that $\gamma \geq \frac{1}{32(1-\alpha)}$.  This is done in Lemma \ref{lem:lowergamma}.  This  and  Lemma~\ref{lem:partiallem} gives a contradiction to the definition of $m^*$.

 %We establish a stronger bound on $\gamma$ in the following lemma establishing that  $\gamma \geq \frac{1}{32(1-\alpha)}$. Hence we will have $\gamma \geq \max\{\frac{1}{32(1-\alpha)}, 2\}$, which and the above lemma proves Theorem~\ref{thm:our-1}.

\begin{lemma}
\label{lem:lowergamma}
 Let $S_1, S_2, \ldots S_{\lceil 2m^*/\alpha \rceil}$ be pairwise disjoint sets of $\alpha$-tight jobs such that $I(S_1) \subseteq I(S_2) \subseteq  \ldots \subseteq  I(S_{\lceil 2m^*/\alpha \rceil}) $.  It is the case that $32(1-\alpha)|I(S_{\lceil 2m^*/\alpha \rceil})| \geq |I(S_1)|$.
\end{lemma}  

\begin{proof}
Suppose that the lemma is not true and $32(1-\alpha)|I(S_{\lceil 2m^*/\alpha \rceil})| < |I(S_1)|$.  To begin, we construct sets $S'_i \subseteq S_i$ such that $I(S'_i) = I(S_i)$ and for each $t \in I(S'_i)$ there are at most two jobs $j \in S'_i$ where $t \in I(j)$.   The construction of such a set is standard in the scheduling community.  A full proof can be found in \cite{ChenMS16}. Such a set can be constructed using a simple greedy procedure where jobs are chosen greedily  such that you always choose to add the job to $S'_i$ from $S_i$ with the latest deadline that covers the smallest uncovered time $t \in I(S_i)$.  

Fix a set $S'_i$.  Partition the jobs in $S'_i$ into two sets $J_{i,1}$ and $J_{i,2}$.  Let $J_{i,1}$ contain a job $j\in S'_i$ if $|I(j) \cap I(S_1)| \geq 4(1-\alpha)|I(j)|$ and otherwise job $j$ is in $J_{i,2}$.   

First say that there exists an $i$ such that $|I(J_{i,2}) \cap I(S_1)| \geq \frac{1}{4}|I(S_1)|$.  Then we have the following.
\begin{eqnarray*}
&&|I(S_{\lceil 2m^*/\alpha \rceil})|\\
 &\geq & |I(S_i)| \quad \quad \quad \quad \mbox{[Since $I(S_i) \subseteq S_{\lceil 2m^*/\alpha \rceil}$]} \\
&\geq&  |I(S'_i)| \quad \quad \quad \quad \mbox{[By definition of $I(S'_i)$]}\\
&\geq&  |I(J_{i,2})| \quad \quad \quad \quad \mbox{[Since $J_{i,2} \subseteq S'_i$]} \\
&\geq&  \frac{1}{2}\sum_{j \in J_{i,2}} |I(j)| \;\;\;\;\quad  \mbox{[Definition of $S'_{i}$]} \\
&\geq&  \frac{1}{ 8(1-\alpha)}\sum_{j \in J_{i,2}} |I(j) \cap I(S_1)| \;\;\;\; \mbox{[Definition of $J_{i,2}$]} \\
&\geq&  \frac{1}{ 32(1-\alpha)}|I(S_1)| \;\;\;\; \\
&& \quad \quad \quad \quad \mbox{[since $|I(J_{i,2}) \cap I(S_1)| \geq \frac{1}{4}|I(S_1)|$]} 
\end{eqnarray*}

We note that the fourth inequality uses the fact that no time is covered by more than two jobs in $S'_i$.  

This contradicts the assumption that $32(1-\alpha)|I(S_{\lceil 2m^*/\alpha \rceil})| < |I(S_1)|$.  Thus, we may assume that there is no $i$ where $|I(J_{i,2}) \cap I(S_1)| \geq \frac{1}{4}|I(S_1)|$.  In particular, since $I(S_i) \subseteq  I(S_i) = I(S'_i) = I(J_{i,1} \cup J_{i,2})$, it is the case that $|I(J_{i,1}) \cap I(S_1)| \geq \frac{3}{4}|I(S_1)|$.

 We will draw a contradiction by showing that the amount of work that must be done during $I(S_1)$ is greater than any feasible schedule can complete using $m^*$ machines. 

Consider any job $j \in \cup_i J_{i,1}$.  Consider the amount of work of job $j$ that must be done during $I(S_1)$ by any feasible schedule.  This is at least $q_j := p_j - (|I(j)|- |I(S_1) \cap I(j)|)$.  Knowing that $j$ is $\alpha$-tight, we have that $q_j \geq \alpha|I(j)| - (|I(j)|- |I(S_1) \cap I(j)|)= |I(S_1) \cap I(j)| - (1-\alpha )|I(j)| $.  

Let $\lambda$ be such that  $\lambda|I(j)| = |I(S_1) \cap I(j)|$.  Then, we have 
$|I(S_1) \cap I(j)| - (1-\alpha )|I(j)| = (1-\frac{1-\alpha}{\lambda})|I(S_1) \cap I(j)|$.
%We have that $|I(S_1) \cap I(j)| - (1-\alpha )|I(j)|  = \lambda |I(j)|- (1-\alpha )|I(j)| = (\lambda - 1 + \alpha)|I(j)| \geq \frac{1}{\lambda}(\lambda - 1 + \alpha)|I(S_1) \cap I(j)| = (1-\frac{1-\alpha}{\lambda})|I(S_1) \cap I(j)|$.  
By definition of $J_{i,1}$ it is the case that $|I(j) \cap I(S_1)| \geq 4(1-\alpha)|I(j)|$ and so $\lambda \geq 4(1-\alpha)$.  Therefore we have that $q_j \geq (1-\frac{1-\alpha}{\lambda})|I(S_1) \cap I(j)| \geq \frac{3}{4}|I(S_1) \cap I(j)|  $.

The argument above gives that each job $j \in \cup_i J_{i,1} $ must be processed for $ \frac{3}{4}|I(S_1) \cap I(j)| $ time units during $I(S_1)$.  The total amount of work  that must be done during $I(S_1)$ for jobs in $J_{i,1}$ is at least the following for any $i$. 

\begin{eqnarray*}
&&\sum_{j \in J_{i,1}}  \frac{3}{4}|I(S_1) \cap I(j)| \\
&\geq& \frac{3}{4} |I(J_{i,1}) \cap I(S_1)| \\
&\geq& \frac{9}{16} | I(S_1)| \;\;\;\; \mbox{[since $|I(J_{i,1}) \cap I(S_1)| \geq \frac{3}{4}|I(S_1)|$]} \ .
\end{eqnarray*}

  There are $\lceil 2m^*/\alpha \rceil$ sets $S_i$ and unique jobs in each set.  Thus, the total volume that must be processed during $I(S_1)$ is greater than  $\lceil 2m^*/\alpha \rceil  \frac{9}{16} | I(S_1)| > m^* |I(S_1)| $.  This is more work than any algorithm with $m^*$ machines can do during $I(S_1)$, contradicting the definition of $m^*$.   
\end{proof}

	%Finally, we briefly explain why the previous work indeed implies the stronger statement, Theorem~\ref{thm:our-1}.

\subsection{The Second Lower Bound}
	\label{sec:lb2}
	
  The authors in \cite{ChenMS16} give another lower bound based on a variant of the definition of a critical pair.
  
  \begin{definition}[\cite{ChenMS16}]
  	\label{def:kevin2}
  	Let $G$ be a set of $\alpha$-tight jobs and let $T$ be a non-empty finite union of time intervals. For some $\mu \in \mathbb{N}$ and $\beta \in (0, 1)$, a pair $(G, T)$ is called \textbf{weakly} $(\mu, \beta)$-critical if
\begin{enumerate}
	\item each time $t$ belonging to an interval in $T$ is covered by at least $\mu$ distinct jobs in $G$.
	\item $|T| \geq \beta / \mu \cdot \sum_{j \in G} \ell_j$.
\end{enumerate}	
  \end{definition}
  
  \begin{theorem}[\cite{ChenMS16}]
	\label{thm:kevin-2}
	If there exists a weakly $(\mu, \beta)$-critical pair, then $m^* = \Omega( \frac{\mu}{\log 1 / \beta})$.
\end{theorem}

This theorem can be strengthened as was done for the first lower bound.  This was shown independently in~\cite{AzarCohen2017}.

\begin{theorem}
	\label{thm:our-2}
	If there exists a weakly $(\mu, \beta)$-critical pair, then $m^* = \Omega( \frac{\mu}{\log_{1/(1 - \alpha)} 1 / \beta})$.
\end{theorem}

	The proof of Theorem \ref{thm:our-2}  extends the proof of Theorem \ref{thm:kevin-2} exactly as the proof of Theorem \ref{thm:our-1} extends the proof of Theorem \ref{thm:kevin-1}.  
	
	The proof in \cite{ChenMS16} shows the following lemma.

\begin{lemma}[\cite{ChenMS16}]
\label{lem:partiallem}
If there exists a pair $(G,T)$ that is weakly $(\mu, \beta)$-critical then there exists a collection $S_1, S_2, \ldots S_{\lceil 2m^*/\alpha \rceil}$ of pairwise disjoint sets of $\alpha$-tight jobs  where $I(S_1) \subseteq I(S_2) \subseteq  \ldots \subseteq  I(S_{\lceil 2m^*/\alpha \rceil}) $.  Further if $|I(S_{\lceil 2m^*/\alpha \rceil})| \geq \gamma |I(S_1)|$ then  $m^* \geq \Omega(\frac{\mu}{\log_{\gamma} \frac{1}{\beta}})$.
\end{lemma}

Combining this lemma with Lemma~\ref{lem:lowergamma} proves Theorem~\ref{thm:our-2}.

\subsection{Analysis of SJF on Jobs with Similar Relative Laxities}
	\label{sec:analysis}

This section is devoted to proving Lemma~\ref{lem:initial}. Fix $\lambda_1$ and $\lambda_2$ such that all jobs have relative laxity in $ [\lambda_1, \lambda_2]$.  For the lemma it is  assumed that  $[\lambda_1, \lambda_2] \subseteq (0, 1/2]$. Consider such a job instance where all jobs are scheduled  using SJF on $m$ machines and $m^*$ is the minimum number of machines required for any algorithm to feasibly schedule this problem instance. That is, at each time jobs are sorted by their original processing times and the $m$ jobs with smallest processing times are processed on the $m$ machines. Let $j$ be the first job SJF couldn't complete before its deadline. Let $T$ be the set of times when job $j$ was not being processed during its lifespan. Let $G$ denote the set of jobs SJF schedules at times in $T$. Note that there are at least $m$ jobs processed by SJF at each time in $T$, meaning that each time in $T$ is covered by at least $m$ distinct jobs in $G$. Thus, $(G, T)$ satisfies the first property in Definition~\ref{def:kevin2} for $\mu = m$. 
	
	It now remains to show the second property in Definition~\ref{def:kevin2}.  To see this we first upper bound the total size of jobs in $G$. We know that every job $i$ in $G$ is 1/2-tight and no larger than job $j$. Thus, $i$'s lifespan length, $|I(i)| \leq 2 p_i \leq 2 p_j \leq 2 |I(j)|$. Since $i$'s lifespan intersects $j$'s lifespan, $i$'s lifespan must be contained in $(r_j - 2 |I(j)|, d_j + 2 |I(j)|)$, implying that the total  length of jobs in $G$ is at most $5 m^* |I(j)|$. This is because the total amount of work any algorithm with $m^*$ machines can do during $(r_j - 2 |I(j)|, d_j + 2 |I(j)|)$ is upper bounded by $5 m^* |I(j)|$.

This implies that $\sum_{i \in G} \ell_i \leq 5 \lambda_2 m^* |I(j)|$ since a job's laxity is at most $\lambda_2$ times its lifespan. Finally, we know that $|T| \geq \ell_j \geq \lambda_1 |I(j)|$ as $j$'s relative laxity is at least $\lambda_1$. Thus we have, $|T| \geq \frac{\lambda_1}{5\lambda_2} \cdot \frac{1}{m^*} \cdot \sum_{i \in G} \ell_i \geq 
\frac{\lambda_1}{5 \lambda_2} \cdot \frac{1}{m} \cdot \sum_{i \in G} \ell_i$; the last inequality follows since our algorithm uses as many machines as the optimal scheduler. This implies that the second property is satisfied for $\beta = \frac{\lambda_1}{5 \lambda_2}$ and $\mu = m$. Hence the pair $(G, T)$ is weakly $(m, \frac{\lambda_1}{5 \lambda_2})$-critical. Since all jobs are $1 - \lambda_2$ tight because the relative laxities are at most $\lambda_2$, by Theorem~\ref{thm:our-2}, we have $m^* \geq c \cdot  \frac{\mu}{\log_{1 / (1 - \alpha)} 1/ \beta} = c \cdot \frac{m}{\log_{1 / \lambda_2} 5 \lambda_2 / \lambda_1}$ for a certain constant $c > 0$. Thus, we have $m \leq \frac{1}{c} \cdot(\log_{1 / \lambda_2} 5 \lambda_2 / \lambda_1) \cdot m^*$. Therefore, we had run SJF on more than $\frac{1}{c} \cdot (\log_{1 / \lambda_2} 5 \lambda_2 / \lambda_1) \cdot m^*$ machines, we would get a contradiction. 
%a contradiction when using $\Theta( (\log_{1 / \lambda_2} \lambda_2/\lambda_1)m^*)$ machines for an appropriate choice of a constant. 
This completes the proof of Lemma~\ref{lem:initial}.

\subsection{Analysis of the CMS Algorithm on Very Tight Jobs and EDF on Loose Jobs}
\label{sec:theoremadapt}

First consider the analysis of CMS on very tight jobs. Observe that, in Theorem ~\ref{thm:our-1},
 the base of the logarithm is $m$.  Further, the analysis of \cite{ChenMS16} gives a $(\mu, 1 / \mu)$-critical pair where $\mu = m'+1$ if the algorithm cannot feasibly schedule all of the jobs on $m'$ machines.  If CMS uses $\Theta(m^*)$ machines then this implies that all jobs can be feasibly scheduled since otherwise the theorem would give a contradiction.

The paper of \cite{ChenMS16} 
shows that for any $\alpha \in (0, 1)$, the algorithm Earliest Deadline First (EDF) is a $1 / (1 - \alpha)^2$-machine optimal when all jobs are $\alpha$-loose. 
Hence EDF will feasibly schedule all $1/2$-loose jobs on $4 m^*$  machines.

\section{Conclusion}

This paper shows that if a given set of jobs can be feasibly scheduled on $m$ machines by any algorithm then there is an online algorithm that will feasibly schedule the jobs on $O(m \log \log m)$ machines.   We point out two exciting open questions remaining in this line of work. One is to reduce the number of machines required to feasibly schedule the jobs to $O(m)$.  The other is to give a feasibility test for the algorithm.  That is, given a task system, determine if the algorithm will feasibly schedule the jobs arising from this task system.

\bibliographystyle{IEEEtran}
\bibliography{onlinemachine}

%\appendix

%
% <OR> manually copy in the resultant .bbl file
% set second argument of \begin to the number of references
% (used to reserve space for the reference number labels box)

%\begin{thebibliography}{1}
%
%\bibitem{IEEEhowto:kopka}
%H.~Kopka and P.~W. Daly, \emph{A Guide to \LaTeX}, 3rd~ed.\hskip 1em plus
%  0.5em minus 0.4em\relax Harlow, England: Addison-Wesley, 1999.
%
%\end{thebibliography}

% that's all folks
\end{document}